\documentclass[a4paper]{article}
\usepackage{amsxtra,amssymb,amsmath,enumerate,graphicx,hyperref,amsthm,color}
\usepackage[T1]{fontenc}
\usepackage[utf8]{inputenc}
\usepackage{lmodern}
\usepackage[english]{babel}
\usepackage{csquotes}
\usepackage[active]{srcltx}
\usepackage{t1enc}
\usepackage{aurical}
\usepackage{float}
\usepackage{subcaption}
\usepackage{titlesec}
\usepackage{float}
\usepackage{accents}
\usepackage{tabularx}
\usepackage{algorithm}
\usepackage[noend]{algpseudocode}
\usepackage[margin=1.in]{geometry}
\usepackage{apacite}

\newtheorem{lemma}{Lemma}[section]
\newtheorem{proposition}[lemma]{Proposition}
\newtheorem{theorem}[lemma]{Theorem}

\newtheorem{definition}[lemma]{Definition}

\newtheorem{assumption}[lemma]{Assumption}

\newcommand{\e}{\mathbb{E}}
\newcommand{\E}{\mathbb{E}}

\renewcommand{\P}{\mathbb{P}}
\newcommand{\R}{\mathbb{R}}

\newcommand{\1}{{\mathbf 1}}

\def\be{\begin{eqnarray}}
\def\ee{\end{eqnarray}}
\def\b*{\begin{eqnarray*}}
\def\e*{\end{eqnarray*}}

\def\be{\begin{eqnarray}}
\def\ee{\end{eqnarray}}
\def\beq{\begin{equation}}
\def\eeq{\end{equation}}
\def\b*{\begin{eqnarray*}}
\def\e*{\end{eqnarray*}}
\def\bi{\begin{itemize}}
\def\ei{\end{itemize}}

\def \1{{\bf 1}}

\def\={\;=\;}

\def\argmin{\mbox{\rm arg}\min}

\def \E{\mathbb{E}}

\def \P{\mathbb{P}}

\def \R{\mathbb{R}}

\def\Ac{{\cal A}}

\def\Fc{{\cal F}}

\def\Lc{{\cal L}}

\def\m{\mathrm{m}}

\def\Lc{{\cal L}}

\def\Qc{{\mathcal Q}}

\def\Vc{{\cal V}}

\newtheorem{Remark}{Remark}[part]

\makeatletter \@addtoreset{equation}{section}

\@addtoreset{Definition}{section}

\@addtoreset{Theorem}{section}

\@addtoreset{Proposition}{section}

\@addtoreset{Assumption}{section}

\@addtoreset{Corollary}{section}

\@addtoreset{Lemma}{section}

\@addtoreset{Remark}{section}

\@addtoreset{Example}{section}

\begin{document}
\title{Variance optimal hedging with application to  Electricity markets}

\author{ Xavier Warin\footnote{EDF R\&D \& FiME, Laboratoire de Finance des March\'es de l'Energie} \thanks{warin@edf.fr} 
}

\maketitle

\begin{abstract}
In this article, we use the mean variance hedging criterion to value contracts   in incomplete markets. Although the problem is well studied in a continuous and even discrete framework, very few works incorporating illiquidity constraints  have been achieved and no algorithm is available in the literature to solve this problem.  We first show that the valuation  problem incorporating  illiquidity constraints  with a mean variance criterion admits a unique solution. Then we develop two  Least Squares Monte Carlo algorithms based on the dynamic programming principle to efficiently value these contracts: in these methods, conditional expectations are classically  calculated by regression using a dynamic programming approach discretizing the control. The first algorithm calculates the optimal value function while the second calculates the optimal cash flow generated along the trajectories.
In a third part, we  give  the  example of the valuation of a load curve contract coming from energy markets. In such contracts,  incompleteness comes from the uncertainty of the customer's load that cannot be hedged and very tight  illiquity constraints are present.  We compare  strategies given by the two algorithms  and  by  some closed formula ignoring constraints and we show that hedging strategies can be very different. A numerical study of the convergence of the algorithms is also given.                                                 
\end{abstract}

\vspace{8mm}


\vspace{8mm}

\section{Introduction} 
%----------------------------------------------------------------------------------------------------------------
Since the deregulation of the energy market in the 1990, spot and future contract on electricity are available
in many countries. Because electricity cannot be stored easily and because  balance between production and consumption has to be checked at every
moment, electricity prices exhibit spikes when there is a high demand or a shortage in production. These high variations in the price lead to fat
tails of the distribution of the log return of the price and Gaussian models for example fail to
reproduce this feature.  Many models have been developed to reproduce the statistical features of electricity prices  as shown in \citeA{aid2015electricity} leading to incomplete market modeling. However even in the case of a Gaussian model   there are many sources of friction and incompleteness on energy market:
 The first is the liquidity of hedging products available on the market.  Large change in position  leads to change in prices and then the hypothesis of an exogenous price model is not valid anymore. In order to limit their impacts on price, risk managers spread in time  their  selling/buying orders. A realistic hedging strategy has to take these constraints into account.
Besides  due to the  bid ask spread which is already a source of friction and due to the inaccuracy of the models used, hedging is only achieved once a week or twice a month so that the hypothesis of a continuous hedge is far from being satisfied. Further, on electricity  market, some contracts are not only dependent on the prices. For example a retailer has to provide the energy to fit a load curve which is uncertain mainly due to its sensitivity to  temperature and   business activity so that a perfect hedge of such contracts is impossible.\\
The literature on the effect of a discrete hedging has been theoretically studied  in complete markets in \citeA{zhang1999couverture}, \citeA{bertsimas2000time}, \citeA{hayashi2005evaluating}. 
In the case of incomplete markets  with some jumps some results about the hedging error due to the  discrete hedging dates can be found in \citeA{tankov2009asymptotic}.\\
The Leland proxy (\citeA{leland1985option}) was the first proposed to take into account discrete hedging dates and transaction costs. One theoretical result obtained using a Black Scholes model gives that when proportional transaction costs decrease with
the number of interventions $n$ proportionally to  $n^{-\alpha}$ with  $\alpha \in ]0, \frac{1}{2}]$,   the hedging error goes to zero in probability when a call is valued and hedged using the Black Scholes formula with a modified volatility.\\
Following this first work, a huge literature on transaction cost has developed. For example, the case $\alpha=0$ corresponding to the realistic case where the transaction cost are independent on the frequency to the hedge has been dealt in \citeA{kabanov1997leland}, \citeA{pergamenshchikov2003limit}, \citeA{denis2010mean}, \citeA{soner1995there}.\\
The case of a limited availability of hedging products has not been dealt theoretically nor numerically so far in the literature to our knowledge. Most of the research currently focuses on  developing some price impact features to model the liquidation of a position (see for example \citeA{gatheral2010no}) instead of imposing some depth limits.\\
In order to  get realistic hedging strategies and because of all the sources of incompleteness,  a risk  criterion has to be used to define the optimal hedging strategy of the contingent claim with a given payoff at maturity date $T$.
In this article we  use a mean variance strategy to define this optimal policy taking into account the fact that hedging is achieved at discrete dates, that transaction costs are present and that the orders are limited in volume at each date.\\
Mean variance hedging is theoretically well studied in the literature.
The first paper on the subject (\citeA{duffie1991mean}) in  a continuous setting was followed by a many articles on the subject.
As we mentioned before, the continuous assumption is not satisfied in practice and an approach taking into account some discrete dates of hedging was proposed in \citeA{follmer1988hedging}. As in the continuous case, some attention has been focus on finding semi-explicit formula for hedging strategies (see \citeA{goutte2014variance} and references inside).\\
In the general case, with transaction costs and constraints on hedging position, there are no explicit solutions so it is necessary to rely on numerical methods.
There are no algorithms proposed to solve this problem:
Mean variance algorithm without friction and based on tree resolution have been developed by \citeA{vcerny2004dynamic} but tree methods are only effective in very low dimension.
The main contribution of this article is to propose two algorithms based on a dynamic programming approach to solve this problem in the general case. Using Monte Carlo methods and relying on regressions  (Least Square Monte Carlo) to calculate the conditional expectations, we propose a first version
calculating the function value recursively in time by discretizing the command and  the amount invested in risky assets.
A second version is proposed tracking the optimal cash flows generated by the strategy using a methodology similar to the one used in  \citeA{longstaff2001valuing}. The use of the Monte Carlo method  permits to have flexible algorithm permitting to solve the problem in multi dimension with respect to the number of hedging products.\\
The article has the following structure:\\
In a first part, in a general setting using the previous work of \citeA{schweizer1995variance}, \citeA{motoczynski2000multidimensional} and \citeA{beutner2007mean}, we show that the problem admits a solution in a general case.\\
In a second part  we develop the two versions of the algorithm.
\\
In a third part, we focus on the problem of an energy retailer that needs to hedge an open position corresponding to a stochastic load curve by trading future contracts.
Using a Gaussian one factor HJM model for the future curve and an Ornstein Ulhenbeck model for the load curve, we  show that on electricity market,
taking into account the discrete hedge and the limited orders has a high impact on the optimal strategy.  We also compare numerically the two versions of the algorithm and study some convergence properties on one test case. \\
In the whole paper we only consider  the case with only one hedging instrument. There is no technical problem to deal with more hedging instruments but numerical experiments are easier with only one asset due to the different discretizations used  in the proposed algorithms.

\section{Mean variance hedging in a general framework}
In this section we suppose that $(\Omega, \Fc, (\Fc_t)_{t \in \left[0,T\right]})$ is a filtered probability space.
We define a set of trading dates $\mathcal{T}= \{t_0=0, t_1, \ldots , t_{N-1}, t_{N}=T \}$ and
we suppose that $(S_t)_{t_0.t_N}$ is  an  almost surely positive hedging product $(S_t)_{t_0.t_N}$,  square integrable so that $\E\left[ S_t^2\right] < \infty$  and adapted so that $S_{t}$ is $\Fc_t$-measurable
for $t=t_0, \ldots, t_N$.\\ 
Further we suppose that the risk free rate is zero so that a bond has always a value of $1$.\\
We suppose that the  contingent claim  $H \in \Lc^2(P)$ is  a $\Fc_T$-measurable random variable.
In the case of a European call option on an asset $S_t$ with strike $K$ and maturity $T$, $H(\omega)= (S_T(\omega)-K)^+$. \\
In this paper we are only interested in self-financing strategies with limited orders, so with bounded controls. Extending \citeA{motoczynski2000multidimensional},  \citeA{beutner2007mean}'s definition, we define:
\begin{definition}
A $(\bar m, \bar l)$ self-financing strategy $\Vc =(\Vc_{t_i})_{i=0,\ldots,N-1}$ is a pair of adapted process $(m_{t_i},l_{t_i})_{i=0, \ldots,N-1}$ defined for
$(\bar m, \bar l) \in (0, \infty) \times (0, \infty)$ such that:
\begin{itemize}
\item   $ 0 \le m_{t_i} \le \bar m,  \quad 0 \le l_{t_i} \le \bar l  \quad \quad    P.a.s. \quad \forall i = 0 ,\ldots,N-1$,
\item   $m_{t_i} l_{t_i} = 0  \quad P.a.s.  \quad \forall i = 0 ,\ldots,N-1$.
\end{itemize}
\end{definition}
In this definition $m_t$ corresponds to the number of shares sold at date $t$, and $l_t$ the number of shares bought at this date.
\begin{Remark}
The strategies defined in  \citeA{motoczynski2000multidimensional} and \citeA{beutner2007mean} do not impose that $m_t l_t= 0$ so a buy and sell control could happen at the same given date.
\end{Remark}

We note $\Theta^{(\bar m, \bar l)}$ the set of $(\bar m, \bar l)$ self-financing strategy and with the notation $\nu= ( m,l)$ for $\nu \in \Theta^{(\bar m, \bar l)}$.\\
We consider a model of proportional cost, so that an investor buying a share at date $t$ pays $(1+\lambda) S_t$ and an investor selling this share  only receives $(1-\lambda) S_t$.
Assuming that there are no transaction costs on the last date $T$, the terminal wealth of an investor with initial wealth $x$ is given by:
\begin{flalign}
x  - \sum_{i=0}^{N-1} (1+\lambda) l_{t_i}S_{t_i} + \sum_{i=0}^{N-1} (1-\lambda) m_{t_i}S_{t_i} + \sum_{i=0}^{N-1} l_{t_i} S_{t_N} -\sum_{i=0}^{N-1} m_{t_i}  S_{t_N}.
\end{flalign}
\begin{Remark}
The transaction costs on the last date $T$ are related to the nature of the contract. In the case of a pure financial contract, the investor  sells the asset and then some transaction costs have to be paid to clear the final position. On energy market for example, the contract is often associated to physical delivery and no special fees are to be paid.
Besides  on these markets,  futures markets are rather illiquid  and the associated  transaction costs are large whereas spot markets are much more liquid  with transaction costs that can be neglected.
\end{Remark}
As in   \citeA{motoczynski2000multidimensional}, \citeA{beutner2007mean}, we define the risk minimal strategy minimizing the $\Lc^2$ risk of the hedged portfolio supposing that the hedger liquidates her entire accumulated position at date $T$:
\begin{definition}
A $(\bar m, \bar l)$ self-financing strategy $\hat \Vc =(\hat m, \hat l)$ is global risk minimizing for the contingent claim $H$ and the initial capital $x$ if:
\begin{flalign}
\label{argminL2}
\hat \Vc = & \argmin_{\Vc = (m, l) \in \Theta^{(\bar m, \bar l)}}  \E \left[(H-x +\sum_{i=0}^{N-1} (1+\lambda) l_{t_i} S_{t_i} - \right.  \nonumber \\
&   \left. \sum_{i=0}^{N-1} (1-\lambda) m_{t_i} S_{t_i} -  \sum_{i=0}^{N-1} l_{t_i} S_{t_N} +\sum_{i=0}^{N-1} m_{t_i}  S_{t_N})^2 \right].
\end{flalign}
\end{definition}

In order to show the existence of  solution to problem  \eqref{argminL2}, we need to make some  assumptions on the prices as in \citeA{motoczynski2000multidimensional}, \citeA{beutner2007mean}:
\begin{assumption}
\label{assump1}
The price process $S=(S_{t_i})_{i=0,\ldots,N}$ is such that a constant $K_1 > 0$ satisfies:
\begin{flalign*}
\E\left[ \frac{S_{t_i}^2}{S_{t_{i-1}}^2} | \Fc_{t_{i-1}}\right] \le K_1,\quad P.a;s., \quad  \forall i = 1,\ldots,N.
\end{flalign*}
\end{assumption}
\begin{assumption}
\label{assump2}
The price process $S=(S_{t_i})_{i=0,\ldots,N}$ is such that a constant $K_2>0$ satisfies:
\begin{flalign*}
\E\left[ \frac{S_{t_{i-1}}^2}{S_{t_{i}}^2} | \Fc_{t_{i-1}}\right] \le K_2, \quad P.a.s., \quad \forall i = 1,\ldots,N.
\end{flalign*}
\end{assumption}
\begin{assumption}
\label{assump3}
The price process $S=(S_{t_i})_{i=0,\ldots,N}$ is such that a constant  $\delta \in (0,1)$ satisfies:
\begin{flalign*}
(\E\left[ S_{t_{i}}| \Fc_{t_{i-1}}\right])^2 \le  \delta E\left[  S_{t_{i}}^2| \Fc_{t_{i-1}}\right] , \quad P.a.s., \quad \forall i = 1,\ldots,N.
\end{flalign*}
\end{assumption}
We introduce   the gain function:
\begin{definition}
For $\Vc \in \Theta^{(\bar m, \bar l)}$,  we define the gain functional $G_T : \Theta^{(\bar m, \bar l)} \longrightarrow \Lc^2$ by
\begin{flalign}
\label{Gain}
 G_T(\Vc) := - \sum_{i=0}^{N-1} (1+ \lambda) l_{t_i} S_{t_i} + \sum_{i=0}^{N-1}l_{t_i} S_T + \sum_{i=0}^{N-1} (1- \lambda) m_{t_i} S_{t_i} - \sum_{i=0}^{N-1} m_{t_i} S_T
\end{flalign}
\end{definition}
In order to show that our problem is well-defined, we want to show that $G_T(\Theta^{(\bar m, \bar l)})$ is closed in $\Lc^2$. As  in \citeA{motoczynski2000multidimensional}, \citeA{beutner2007mean}, we have to introduce another functional $\tilde G_T(\Vc) : \Theta^{(\bar m, \bar l)} \longrightarrow \Lc^2 \times \Lc^2 \times  \Lc^2 $ defined by:
\begin{flalign}
\label{GainSep}
 \tilde G_T(\Vc) := \left( \begin{array}{l}
- \sum_{i=0}^{N-1} (1+ \lambda) l_{t_i} S_{t_i}  \\ \sum_{i=0}^{N-1} (1- \lambda) m_{t_i} S_{t_i} \\  \sum_{i=0}^{N-1} (l_{t_i} -m_{t_i}) S_T.
\end{array} \right)
\end{flalign}
\begin{proposition}
\label{PropCloseSep}
Under assumptions \ref{assump1},\ref{assump2},\ref{assump3},  $\tilde G_T(\Theta^{(\bar m, \bar l)})$ is a closed bounded convex set of $\Lc^2 \times \Lc^2 \times \Lc^2$.
\end{proposition}
\begin{proof}
The convexity of the set is  due to the linearity of the $\tilde G_T$ operator. Because of the boundedness of $(l_{t_i}, m_{t_i})$, $i=0, N-1$ and the $(S_{t_i})_{i=0,N}$ square integrability, we have that $\tilde G_T(\Theta^{(\bar m, \bar l)})$ is bounded in $\Lc^2 \times \Lc^2 \times \Lc^2 $.\\
In order to prove the closedness of $\tilde G_T(\Theta^{(\bar m, \bar l)})$, we suppose that we have a sequence  $ (\Vc^n) = (m^n_{t_i}, l^n_{t_i})_{i=0, N-1} \in  \Theta^{(\bar m, \bar l)}$ such that $\tilde G_T(\Vc^n)$  converges in $\Lc^2 \times \Lc^2 \times \Lc^2 $. \\
Using assumption  \ref{assump3} and lemma 5.3 in \citeA{motoczynski2000multidimensional}, we get that for $ i= 0,\ldots, N-1$, there exist $l_{t_i}^\infty$, $m_{t_i}^\infty$ $\Fc_{t_i}$-adapted such that $l_{t_i}^\infty S_{t_i} \in \Lc^2$, $m_{t_i}^\infty S_{t_i} \in \Lc^2$ and
\begin{flalign*}
l_{t_i}^n S_{t_i}  \xrightarrow{\Lc^2} l_{t_i}^\infty S_{t_i},\\
m_{t_i}^n S_{t_i}  \xrightarrow{\Lc^2}  m_{t_i}^\infty S_{t_i},
\end{flalign*}
so that
\begin{flalign*}
-  \sum_{i=0}^{N-1} (1+ \lambda) l_{t_i}^n S_{t_i} \xrightarrow{\Lc^2} -  \sum_{i=0}^{N-1} (1+ \lambda) l_{t_i}^\infty S_{t_i}, \\
 \sum_{i=0}^{N-1} (1- \lambda) m_{t_i}^n S_{t_i} \xrightarrow{\Lc^2} \sum_{i=0}^{N-1} (1- \lambda) m_{t_i}^\infty S_{t_i}.
\end{flalign*}
Using Lemma 5 in \citeA{beutner2007mean} based on assumption\ref{assump1}, we get that $l^n_{t_i} S_{T}$ converges weakly to  $l^\infty_{t_i} S_{T}$, and
 $m^n_{t_i} S_{T}$ converges weakly to  $m^\infty_{t_i} S_{T}$,  for $i=0,\ldots,N-1$ such that
 \begin{flalign*}
 \sum_{i=0}^{N-1} (l_{t_i}^n -m_{t_i}^n) S_T \xrightarrow{\text{weak}}   \sum_{i=0}^{N-1} (l_{t_i}^\infty -m_{t_i}^\infty) S_T.
\end{flalign*}
Using the fact that strong and weak limit coincide in $\Lc^2$, we get that the limit of $\tilde G_T(\Vc^n)$ is $\tilde G_T( \Vc^\infty)$ with $\Vc^\infty = (m^\infty, l^\infty) \in \Lc^2 \times \Lc^2$.\\
We still have to show that $l_{t_i}^\infty$, $m^\infty_{t_i}$ respect the constraints:
First using assumption\ref{assump2} and the tower property of conditional expectation, we have
\begin{flalign}
\label{unSurS2Bound}
\E\left[ \frac{1}{S_{t_{i}}^2} \right] \le \frac{K_2^{i}}{S_0^2} < \infty, \forall i=0,\ldots, N.
\end{flalign}
Suppose for example that there exist $\epsilon > 0$ and $i <N$ such that $l_{t_i}^\infty(\omega) > \bar l +\epsilon$ for $\omega$ in a set $\Ac$  with a non zero measure $\tilde \Ac$.
Then, 
\begin{flalign*}
 (\bar l + \epsilon) \tilde \Ac  \le& E\left[ l^\infty_{t_i} 1_{\Ac} \right], \nonumber\\
                     = &E\left[ (l^\infty_{t_i}- l^p_{t_i}) 1_{\Ac}\right] + E\left[ l^p_{t_i} 1_{\Ac}\right], \nonumber \\
                     \le&  E\left[ (l^\infty_{t_i}S_T- l^p_{t_i}S_T) 1_{\Ac}  \frac{1}{S_T}\right] + \bar l \tilde \Ac.
\end{flalign*}
Then using \eqref{unSurS2Bound} we know that $1_{\Ac}  \frac{1}{S_T}$ is in $\Lc^2$ so using  the weak convergence of  $l^p_{t_i}S_T$ towards $l^\infty_{t_i}S_T$
and letting $p$ go to infinity gives the contradiction.\\
The same type of calculation   shows that $l^\infty_{t_i} \ge 0$, $ 0 \le \m^\infty_{t_i} \le \bar m$.\\
At last it remains to show that $l^\infty m^\infty =0$. Suppose for example that there exist $\epsilon > 0$ and $i <N$ such that $l_{t_i}^\infty(\omega) m_{t_i}^\infty(\omega)  > \epsilon$ for $\omega$ in a set $\Ac$  with a non zero measure $\tilde \Ac$.\\
Using the fact that $l_{t_i}^nm_{t_i}^n =0$,  $\frac{m_{t_i}^n 1_{\Ac}}{S_{t_i}} \in \Lc^2$ and Cauchy Schwartz:
\begin{flalign*}
  \epsilon \tilde \Ac  \le & E\left[ l_{t_i}^\infty m_{t_i}^\infty  1_{\Ac} \right], \\
= &  \E\left[(S_{t_i} m_{t_i}^\infty - S_{t_i} m_{t_i}^n) \frac{l_{t_i}^\infty 1_{\Ac}}{S_i} \right] + \E\left[(S_{t_i}l_{t_i}^n- S_{t_i}l_{t_i}^\infty)  \frac{m_{t_i}^n 1_{\Ac}}{S_{t_i}} \right] \\
\le & |\E\left[(S_{t_i} m_{t_i}^\infty - S_{t_i} m_{t_i}^n) \frac{l_{t_i}^\infty 1_{\Ac}}{S_i} \right]| + ||S_{t_i}l_{t_i}^n- S_{t_i}l_{t_i}^\infty||_{\Lc^2}  \bar m   \tilde \Ac  \frac{K_2^{0.5 i}}{S_0}.
\end{flalign*}
Letting $n$ go to infinity, and using the strong/weak convergence property we get the contradiction.
\end{proof}
We are now in position to give the  existence and uniqueness theorem:
\begin{theorem}
Under assumptions \ref{assump1}, \ref{assump2}, \ref{assump3}, the minimization problem  \eqref{argminL2} admits a unique solution.
\end{theorem} 
\begin{proof}
First we show  that  $G_T(\Theta^{(\bar m, \bar l)})$ is a closed bounded convex set in $\Lc^2$. Convexity and boundedness are straightforward.
Suppose that  we have a sequence  $ (\Vc^n) = (m^n_{t_i}, l^n_{t_i})_{i=0, N-1} \in  \Theta^{(\bar m, \bar l)}$ such that $G(\Vc^n)$  converges in $\Lc^2$. 
 $\tilde G_T(\Theta^{(\bar m, \bar l)})$  being a bounded closed convex set due to proposition \ref{PropCloseSep}, it is weakly closed, and we can
extract a sub sequence  $ \Vc^p$  such that $\tilde G_T(\Vc^p)$ converges weakly to $\tilde G_T(\Vc^\infty)$ with $\Vc^\infty \in \Theta^{(\bar m, \bar l)}$.
Then by linearity $G_T(\Vc^p)$ converges weakly to $G_T(\Vc^\infty)$ and due to the fact that weak and strong limit coincide, we get that $G_T(\Vc^p)$ converges strongly to $G_T(\Vc^\infty)$ giving the closedness of $G_T(\Theta^{(\bar m, \bar l)})$.\\
As a result,  the solution of the minimization problem exists, is  unique  and  corresponds to the projection of $H$ on $G_T(\Theta^{(\bar m, \bar l)})$  in $\Lc^2$.
\end{proof}

\section{Algorithms for mean variance hedging}
In this section we suppose that the process is Markov and that the payoff $H$ is a function of the asset value at maturity only to simplify the presentation for the Monte Carlo methods proposed. In a first section  we give two algorithms  based on dynamic programming and regressions to calculate the solution of the problem  \eqref{argminL2}. In a second part we detail  the resolution  for one algorithm using some specific regressors and discretizing commands and stocks.
\subsection{Some dynamic programming algorithms}
We introduce the global position $\nu= (\nu_i)_{i=0, \ldots, N-1}$ with:
\begin{flalign*}
\nu_{i} = \sum_{j=0}^i (m_{t_j} - l_{t_j}), \forall  i=0,\ldots, N-1.
\end{flalign*}
Using the property that $m_{t_i}l_{t_i} =0$ for all  $ i=0, \ldots, N-1$, we get that $|\nu_{i}-\nu_{i-1}| = l_{t_i} +m_{t_i}$  with the convention that $\nu_{-1} =0$ so that the gain functional can  rewritten as:
\begin{flalign*}
G_T(\Vc) = \hat G_T(\nu) =  x - \sum_{i=0}^{N-1} \lambda |\Delta \nu_{i-1} |  S_{t_i} + \sum_{i=0}^{N-1} \nu_{i} \Delta S_i,
\end{flalign*}
where $\Delta S_i = S_{t_{i+1}} -S_{t_i}$, $\Delta \nu_i= \nu_{i+1} -\nu_{i}$.\\
We then introduce $\hat \Theta^{(\bar m, \bar l)}$ the set of  adapted random variable $(\nu_{i})_{i=0,\ldots,N-1}$ such that
\begin{flalign*}
 - \bar m \le \nu_{i}- \nu_{i-1} \le \bar l, \quad  \forall i=1,\ldots,N-1.
\end{flalign*}
The  problem \eqref{argminL2} can be rewritten  as done in \citeA{schweizer1995variance} finding $\hat \nu=(\hat \nu_i)_{i=0, \ldots, N-1}$ satisfying:
\begin{flalign}
\label{argminL2Bis}
\hat \nu = & \argmin_{\nu \in \hat \Theta^{(\bar m, \bar l)}}  \E\left[\big(H- x- \hat G_T(\nu)\big)^2 \right].
\end{flalign}
We introduce the spaces $K_i$,  $i=0,\ldots,N$ of the $\Fc_{t_i}$-measurable and square integrable random  variables.
We define for  $i \in 0,\ldots,N$, $V_i \in K_i$   as:
\begin{flalign}
\label{VDef}
V_{N} =&  H, \nonumber \\
V_{i} =& \E\left[ H - \sum_{j=i}^{N-1}  \nu_{j}  \Delta S_{j} +  \lambda \sum_{j=i}^{N-1} | \Delta \nu_{j-1}|  S_{t_j} \quad |\Fc_{t_i}\right], \forall i =0,\ldots,N-1.
\end{flalign}
Then we can prove the following proposition:
\begin{proposition}
  The problem  \eqref{argminL2Bis} can be rewritten as:
  \begin{flalign}
\label{argminTer}
\hat \nu = & \argmin_{ \nu \in \hat \Theta^{(\bar m, \bar l)}}   \E\left[ \left(V_{N} - \nu_{N-1} \Delta S_{N-1} +  \lambda |\Delta \nu_{N-2}| S_{t_{N-1}}- V_{N-1}\right)^2\right] + \nonumber  \\
& \sum_{i=2}^{N-1} \E\left[ \left(V_{i} +   \lambda  |\Delta \nu_{i-2}| S_{t_{i-1}} - \nu_{i-1} \Delta S_{i-1} -V_{i-1}\right)^2\right] + \nonumber \\
& \E\left[\left(V_{1} +  \lambda |\nu_{0}| S_{t_0} - \nu_{0}  \Delta S_0- x \right)^2\right].
\end{flalign} 
\end{proposition}
\begin{proof}
\begin{flalign}
\label{eq:globMin}
\E\left[(H- x- \hat G_T(\nu))^2 \right] = & E\left[ \big( \left(V_{N} - \nu_{N-1} \Delta S_{N-1}  +  \lambda |\Delta \nu_{N-2}| S_{t_{N-1}} - V_{N-1} \right)+ \right. \nonumber \\ 
& \sum_{i=2}^{N-1} \left(V_i +  \lambda | \Delta \nu_{i-2}| S_{t_{i-1}}  - \nu_{i-1} \Delta S_{i-1} - V_{i-1} \right) +  \nonumber\\
&  \left. \left( V_{1} + \lambda |\nu_0| S_{t_{0}}  - \nu_{0}  \Delta S_0  -x\right) \big)^2\right] 
\end{flalign}
Due to the  definition \eqref{VDef}, we have that
\begin{flalign}
\label{eq:Balance}
 E\left[  V_{i} +  \lambda |\Delta \nu_{i-2} | S_{t_{i-1}} - \nu_{i-1}  \Delta S_{i-1}-V_{i-1}  |\Fc_{t_{i-1}}\right] =0, \forall i =1,\ldots,N,
\end{flalign}
so that
\begin{flalign*}
\E\left[(H- x- \hat G_T(\nu))^2 \right] = & \E\left[ \E\left[ \left(V_N - \nu_{N-1} \Delta S_{N-1}  +  \lambda |\Delta \nu_{N-2}| S_{t_{N-1}}- V_{N-1}\right)^2 | \Fc_{t_{N-1}}\right]\right]+ \\
 & \E\left[ \sum_{i=2}^{N-1} \left(V_{i} +  \lambda |\Delta \nu_{i-2}| S_{t_{i-1}} - \nu_{i-1} \Delta S_{i-1} -V_{i-1} \right)^2 + \right.\\
  & \left.  \left( V_{1} +  \lambda |\nu_{0}| S_{t_{0}}  - \nu_{0} \Delta S_0 - x  \right)^2\right]
\end{flalign*}
and iterating the process gives
\begin{flalign*}
\E\left[(H- x- \hat G_T(\nu))^2 \right] = &  \E\left[ \left(V_{N} - \nu_{N-1}  \Delta S_{N-1} +  \lambda |\Delta \nu_{N-2}| S_{t_{N-1}} - V_{N-1}\right)^2\right] + \\
& \sum_{i=2}^{N-1} \E\left[ \left(V_{i} +  \lambda |\Delta \nu_{i-2}| S_{t_{i-1}} - \nu_{i-1} \Delta S_{i-1} -V_{i-1}\right)^2\right] + \\
& \E\left[\left(V_{1} +  \lambda | \nu_{0}| S_{t_0} - \nu_{0}  \Delta S_0 - x \right)^2\right],
\end{flalign*}
\end{proof}
which gives the result.

We introduce the space
\begin{flalign*}
W_i^{\bar m, \bar l}(\eta) = & \{ (V, \nu) / 
V , \nu \mbox{ are } \R \mbox{ valued } \Fc_{t_i}\mbox{-adapted with }  - \bar m \le  \nu - \eta  \le \bar l \},
\end{flalign*}
and the space
\begin{flalign*}
\hat W_i^{\bar m, \bar l}(\eta) = & \{ (V, \nu_i,\ldots, \nu_{N-1}) / 
 V  \mbox{ is } \R \mbox{ valued, } \Fc_{t_i}\mbox{-adapted , the }  \nu_j , j \ge i \mbox{ are } \R \mbox{ valued } \\
 & \Fc_{t_j}\mbox{-adapted with }   \bar m \le  \nu_i - \eta  \le \bar l   , \bar m \le  \nu_{j+1} - \nu_j  \le \bar l \mbox{ for }   i \le j < N-1\}.
\end{flalign*}
The formulation \eqref{argminTer} can be used to solve the problem by dynamic programming when the price process is Markov and the payoff a function of the asset value at maturity.\\
We introduce the optimal residual $R$ at date $t_i$, for a current price $S_{t_i}$ and having in  portfolio an investment in $\nu_{i-1}$ assets:
\begin{align}
\label{eq:RArgmin}
R(t_i, S_{t_i}, \nu_{i-1}) =  \min_{ (V , \nu) \in  \hat W_i^{\bar m, \bar l}(\nu_{i-1}) }\E\left[ \big (H - \sum_{j=i}^{N-1}  \nu_{j}  \Delta S_{j} +  \lambda \sum_{j=i}^{N-1} | \Delta \nu_{j-1}|  S_{t_j} -V  \big)^2 \quad |\Fc_{t_i}\right].
\end{align}
Then equation \eqref{argminTer} gives:
\begin{align}
R(t_i, S_{t_i}, \nu_{i-1}) =  & \min_{ (V , \nu) \in  W_i^{\bar m, \bar l}(\nu_{i-1}) }\E\left[ \big ( V - \nu_{i}  \Delta S_{i}  + \lambda  | \Delta \nu_{i-1}|  S_{t_i}  - \tilde V \big)^2  + R(t_{i+1}, S_{t_{i+1}},\nu  ) |\Fc_{t_i}\right],
\end{align}
where $\tilde V$  is the first component of the argmin in equation \eqref{eq:RArgmin} calculating $R(t_{i+1}, S_{t_{i+1}},\nu  )$.\\
The underlying recursion gives the corresponding algorithm \ref{algoMeanVar2} where the optimal initial value of cash needed to minimize the future mean variance risk  at the date $t_{i}$ with an asset value $S_{t_i}$  for an  investment $\nu_{i-1}$ chosen at date $t_{i-1}$ is noted $\tilde V(t_i,S_{t_i}, \nu_{i-1})$.
\begin{algorithm}[H]
\caption{\label{algoMeanVar2}Backward resolution for $\Lc^2$ minimization problem with iterated conditional expectation approximation.}
\begin{algorithmic}[1]
\State $  \begin{array}{ll}
  \tilde V(t_{N}, S_{t_{N}}(\omega), \nu_{N-1})  = H(\omega), \quad \forall \nu_{N-1} 
\end{array} $
\State  $    R(t_{N}, S_{t_{N}}(\omega), \nu_{N-1})  =  0, \quad \forall \nu_{N-1}$
\For{$i = N, 2$}
\State \be  
\label{eq:recur}
\begin{array}{ll}
(\tilde V(t_{i-1}, S_{t_{i-1}}, \nu_{i-2}),\nu_{i-1} ) = & \argmin_{ (V,\nu) \in W_{i-1}^{\bar m, \bar l}(\nu_{i-2})}  \E\left[ (  \tilde V(t_i, S_{t_i}, \nu) - \right. \\ 
 & \nu  \Delta S_{i-1}
+ 
\lambda |\nu-\nu_{i-2}| S_{t_{i-1}} -   V )^ 2 +  \\
&  \left. R(t_i, S_{t_i}, \nu) | \Fc_{t_{i-1}}\right]
\end{array} 
\ee
\State $  \begin{array}{ll}  R(t_{i-1}, S_{t_{i-1}}, \nu_{i-2}) = & \E\left[ ( \tilde V(t_i, S_{t_i}, \nu_{i-1}) - \nu_{i-1}  \Delta S_{i-1}
+ \right. \\
&  \lambda |\Delta \nu_{i-2}| S_{t_{i-1}} -   \tilde V(t_{i-1}, S_{t_{i-1}}, \nu_{i-2} ))^ 2 +  \\
&  \left. R(t_i, S_{t_i}, \nu_{i-1})| \Fc_{t_{i-1}}\right]
\end{array} $
\EndFor
\State $ \begin{array}{ll}
\nu_{0} = & \argmin_{\nu \in \left[-\bar m,\bar l\right]}  \E\left[ (V(t_1,S_{t_1}, \nu) + \lambda| \nu| S_{t_0} - \nu \Delta S_{0}- x)^2 +   R(t_{1}, S_{t_{1}}, \nu) \right]
\end{array} $
\end{algorithmic}
\end{algorithm}
Note that  at date $t_i$, the optimal investment strategy
$\eta^*$ calculated  by equation \eqref{eq:recur} is a function of $S_{t_i}$, $\nu_{i-1}$.

\begin{Remark}
The previous algorithm can be be easily modified to solve local risk minimization problem of  Schweizer (see \citeA{schweizer1999guided}) where some liquidity constraints have been added. The equation \eqref{eq:recur} has to be modified by
\begin{flalign*}
(\tilde V(t_{i-1}, S_{t_{i-1}}, \nu_{i-2}),\nu_{i-1} ) = & \argmin_{ (V,\nu) \in W_{i-1}^{\bar m, \bar l}(\nu_{i-2})}  \E\left[ ( \tilde V(t_i, S_{t_i}, \nu_{i-1}) - \right. \\ 
 & \left. \nu  \Delta S_{i-1}
+ 
\lambda |\nu-\nu_{i-2}| S_{t_{i-1}} -   V )^ 2 | \Fc_{t_{i-1}}\right].
\end{flalign*} 
\end{Remark}
Due to approximation errors linked to the methodology used to estimate conditional expectation, the $R$ estimation in algorithm \ref{algoMeanVar2} may be prone to an  accumulation of errors during
the time iterations.
Similarly to the scheme introduced in  \citeA{bender2007forward} to improve the methodology proposed in \citeA{gobet2005regression} to solve Backward Stochastic Differential Equations, we can propose a second version  of the previous algorithm that updates  the gain functional $\bar R$  $\omega$ by $\omega$. Then $\bar R$ satisfies at date $t_{i}$ with an asset value $S_{t_i}$  for an  investment $\nu_{i-1}$ chosen at date $t_{i-1}$:
\begin{align*}
\bar R(t_i, S_{t_i}, \nu_{i-1}) = & H - \sum_{j=i}^{N-1}  \nu_{j}  \Delta S_{j} +  \lambda \sum_{j=i}^{N-1} | \Delta \nu_{j-1}|  S_{t_j},\\
& = \bar R(t_{i+1}, S_{t_{i+1}}, \nu_{i}) -\nu_{i}  \Delta S_{i} +  \lambda | \Delta \nu_{i-1}|  S_{t_i}, 
\end{align*}
and at the date $t_i$ according to equation \eqref{eq:globMin}  the  optimal control  is the control $\nu$ associated to the minimization  problem:
\begin{align*}
\min_{(V,\nu) \in W_{i}^{\bar m, \bar l}(\nu_{i-1})} \E\left[ (\bar R(t_{i+1}, S_{t_{i+1}}, \nu) -  \nu  \Delta S_{i}
+ \lambda |\nu-\nu_{i-1}| S_{t_{i}} - V )^ 2 | \Fc_{t_{i}}\right].
\end{align*}
This leads to the second algorithm \ref{algoMeanVar3}.
\begin{algorithm}[H]
\caption{\label{algoMeanVar3}Backward resolution for $\Lc^2$ minimization problem avoiding conditional expectation iteration.}
\begin{algorithmic}[1]
\State  $ \bar R(t_{N}, S_{t_{N-1}}(\omega), \nu_{N-1})  =  H(\omega),  \quad \forall \nu_{N-1}$
\For{$i = N, 2$}
\State   \begin{eqnarray}
  (\tilde V(t_{i-1}, S_{t_{i-1}}, \nu_{i-2}),\nu_{i-1}) = & \argmin_{ (V,\nu) \in W_{i-1}^{\bar m, \bar l}(\nu_{i-2})}  \E\left[ (\bar R(t_i, S_{t_i}, \nu) - \right. \nonumber \\
    & \left. \nu  \Delta S_{i-1}
+ \lambda |\nu-\nu_{i-2}| S_{t_{i-1}} - V )^ 2 | \Fc_{t_{i-1}}\right]
\label{argminBack} \end{eqnarray}
\State $  \begin{array}{ll}  \bar R(t_{i-1}, S_{t_{i-1}}(\omega), \nu_{i-2}(\omega)) = &  \bar R(t_{i}, S_{t_{i}}(\omega), \nu_{i-1}(\omega)) - \nu_{i-1}(\omega)   \Delta S_{i-1}(\omega)
+  \lambda |\Delta \nu_{i-2}(\omega)| S_{t_{i-1}}(\omega)
\end{array} $
\EndFor
\State $ \begin{array}{ll}
\nu_{0} = & \argmin_{\nu \in \left[-\bar m,\bar l\right]}  \E\left[ ( \bar R(t_1, S_{t_1}, \nu) + \lambda| \nu| S_{t_0} - \nu \Delta S_{0} - x)^2  \right]
\end{array} $
\end{algorithmic}
\end{algorithm}

\begin{Remark}
\label{remarkMeanVar}
In order to treat the case of mean variance hedging that consists in finding the optimal strategy and the initial wealth to hedge the contingent claim the last line of algorithm \ref{algoMeanVar2} is replaced by
\begin{flalign*}
(\tilde V,\nu_{0}) = & \argmin_{(V,\nu)}  \E\left[ (V(t_1,S_{t_1}, \nu) + \lambda |\nu| S_{t_0} - \nu \Delta S_{0} - V)^2 + R(t_{1}, S_{t_{1}}, \nu) \right],
\end{flalign*}
and last line of algorithm \ref{algoMeanVar3} by
\begin{flalign*}
(\tilde V,\nu_{0}) =  \argmin_{ (V,\nu) \in \R \times \left[-\bar m,\bar l\right]}  \E\left[ ( \bar R(t_1, S_{t_1}, \nu) + \lambda | \nu| S_{t_0} - \nu \Delta S_{0} - V)^2  \right].
\end{flalign*}
\end{Remark}
\begin{Remark}
In the two algorithms presented an argmin has to be achieved:  a discretization in $\nu_{i-2}$ has to be achieved on a grid $\left[\nu_{i-1}-m,\nu_{i-1}+l\right]$.
\end{Remark}
\subsection{Effective implementation based on local regressions.}\label{sec:reg}
Starting from the theoretical algorithms  \ref{algoMeanVar2}, \ref{algoMeanVar3}  we aim at getting an  effective implementation based on  a representation of the function $\tilde V$ depending on time, $S_t$ and the position $\nu_t$ in the hedging assets. We only give a detailed implementation for the algorithm \ref{algoMeanVar3} but an adaptation is straightforward for the algorithm \ref{algoMeanVar2}.
To simplify the setting, we suppose that only one hedging product is available.
The extension in the multi-dimensional case is straightforward.
\begin{itemize}
    \item In order to represent the dependency in the hedging position we  introduce a time dependent grid  
\begin{flalign*}
\Qc_{i} :=  (\xi k)_{k=- (i+1)\lfloor \frac{\bar m}{\xi} \rfloor,\ldots, (i+1)\lfloor \frac{\bar l}{\xi} \rfloor} 
\end{flalign*}
where $\xi$ is the mesh size associated to the set of grids $(\Qc_i)_{i=0,N}$ and, if possible, chosen such that $\frac{\bar l}{\xi} =\lfloor \frac{\bar l}{\xi} \rfloor$ and $\frac{\bar m}{\xi}= \lfloor \frac{\bar m}{\xi} \rfloor$.
\item To represent the dependency in $S_t$ we  use a Monte Carlo method using simulated path $\left( (S_{t_i}^{(j)})_{i=0,\ldots, N}\right)_{j=1,\ldots, M}$ and calculate the $\argmin$ in equation \eqref{argminBack}  using a methodology close to the one described in \citeA{bouchard2012monte}:  
suppose that we are given at each date $t_i$ $(D^i_q)_{q=1,\ldots,Q}$ a partition of $\left[ \min_{j=1,M} S_{t_i}^{(j)}, \max_{j=1,M} S_{t_i}^{(j)}\right]$  such that each cell contains the same number of samples. We use the $Q$ cells  $(D^i_q)_{q=1,\ldots,Q}$ to represent the dependency of $\tilde V$ in the $S_{t_i}$ variable.\\
On each cell $q$ we search for $\hat V^q$ a linear approximation of the function $\tilde V$ at a given date $t_i$ and for a position $k \xi$ so that
$\hat V^q(t_i,  S, k) = a_i^q + b_i^q S$ is an approximation of $ \tilde V(t_i,S, k \xi)$.
To find the optimal command on each path by arbitrage,  conditional variance has to be calculated also by regression.
\end{itemize}
Let us note   $(l_i^q(j))_{j=1,\frac{M}{Q}}$  the set of all samples belonging to the cell $q$ at date $t_i$. On each mesh, for a current position $k \xi$ in the stock of assets, the optimal control $\hat \nu^q$ is obtained by discretizing the
command $\nu$ on a grid $\eta = ( (k+r) \xi  )_{r = - \lfloor \frac{\bar m}{\xi} \rfloor, \ldots, \lfloor \frac{\bar l}{\xi} \rfloor}$   and by testing the one giving a $\hat V^q$ value minimizing the $\Lc^2$ risk so solving equation \eqref{argminBack}.\\
The algorithm \ref{testCommand} permits to find the optimal $\nu_i^{(j)}(k)$ command  using algorithm \ref{algoMeanVar3} at  date $t_i$, for a hedging position $k \xi$ and for all the Monte Carlo simulations $j$. 
\begin{algorithm}[H]
\caption{\label{testCommand} Optimize minimal hedging position $(\hat \nu_{t_{i}}^{(l)}(k))_{l=1,\ldots,M}$ at date $t_{i-1}$}
\begin{algorithmic}[1]
\Procedure{OptimalControl}{ $\bar R(t_{i+1},.,.) , k, S_{t_i}, S_{t_{i+1}}$}
\For{$q=1, Q$}
\State $P=\infty$, 
\For{$l = -\lfloor \frac{\bar m}{\xi} \rfloor, \ldots,   \lfloor \frac{\bar l}{\xi} \rfloor $}
\State $ \begin{array}{ll}
(a_i^{q,k},b_i^{q,k} ) = \argmin_{ (a,b )\in \R^2}  &  \displaystyle{\sum_{j=1}^{\frac{M}{Q}}}  \big( \bar R(t_{i+1}, S_{t_{i+1}}^{l^q_i(j)}, (k+l) \xi)  - \\& (k+l) \xi \Delta S_{i}^{l^q_i(j)}
+   \\
&  \lambda |l \xi| S_{t_{i}}^{l^q_i(j)} - (a+b S_{t_{i}}^{l^q_i(j)})  \big)^ 2 
\end{array}
$ \Comment{Estimate the value function}
%\State $\begin{array}{ll} \tilde P = &  \displaystyle{\sum_{j=1}^{\frac{M}{Q}}}  \big( \bar R(t_{i+1}, S_{t_{i+1}}^{l^q_i(j)}, (k+l) \xi)  - (k+l) \xi \Delta S_{i}^{l^q_i(j)}
%+   \\ &  \lambda |l \xi| S_{t_{i}}^{l^q_i(j)} -  (a_i^q+ b_i^q S_{t_{i}}^{l^q_i(j)}) \big)^2 \end{array}$
\State $ \begin{array}{ll} (c_i^{q,k},d_i^{q,k} )  =  \argmin_{ (c,d )\in \R^2}  &  \displaystyle{\sum_{j=1}^{\frac{M}{Q}}}   \big( ( \bar R(t_{i+1}, S_{t_{i+1}}^{l^q_i(j)}, (k+l) \xi)  - \\& (k+l) \xi \Delta S_{i}^{l^q_i(j)}
+   \\
&  \lambda |l \xi| S_{t_{i}}^{l^q_i(j)} - (a_i^{q,k}+b_i^{q,k} S_{t_{i}}^{l^q_i(j)}) )^2 -  \\
&(c+ d S_{t_{i}}^{l^q_i(j)}) \big)^2  \end{array}$ \Comment{Conditional variance}
%\If{ $ \tilde P < P$}
%\State $\nu^{q} = k \xi$, $P =  \tilde P$
%\EndIf
\For{$j=1,\frac{M}{Q}$}
\If{$c_i^{q,k} + d_i^{q,k} S_{t_{i}}^{l^q_i(j)} < P(j)$} \Comment{Use conditional variance for arbitrage}
\State $\hat \nu_{i}^{(l^q_i(j))}(k) =  l \xi$, $P(j) = c_i^{q,k} + d_i^{q,k} S_{t_{i}}^{l^q_i(j)}$
\EndIf
\EndFor
\EndFor
\EndFor
\Return $( \hat \nu_{t_{i}}^{(j)}(k))_{j=1,\ldots,M}$ 
\EndProcedure
\end{algorithmic}
\end{algorithm}
\begin{Remark}
The algorithm is based on a discretization of the current hedging position and a discretization of the  admissible hedging positions potentially taken at the next date. So this approach has to face the curse of dimensionality and the number of hedging product shouldn't exceed two. Solving the problem with three hedging assets is possible using some libraries such as the StOpt Library  \citeA{gevret2016stochastic} on a cluster of CPU.
\end{Remark}
\begin{Remark}
 It is also possible to use  global polynomials in regressions as in \citeA{longstaff2001valuing} or to use kernel regression methods as in the recent paper \citeA{langrene2017fast}. A recent comparison of some regression techniques to value some gas storages can be found in \citeA{ludkovski2018simulation}.
\end{Remark}
\begin{Remark}
It is possible to use different discretization $\xi$ to define the set $\eta$ and the set $\Qc_{i}$. Then  grids used for the state and the control may not correspond. In order to save time, the grid for the state can be chosen coarser than the grid for the control.
An example of the use of such an interpolation for gas storage problem tracking the optimal cash flow generated along the Monte Carlo strategies can be found in \citeA{warin2012gas}.
\end{Remark}
\begin{Remark}
This algorithm permits to add some global constraints on the global liquidity of the hedging asset. This is achieved by restricting the possible hedging positions to a subset of $\Qc_{i}$ at each date $t_i$.
\end{Remark}

Then the global discretized version of algorithm \ref{algoMeanVar3} is given on algorithm \ref{algoMeanVar4} where $H^{(j)}$ correspond to the $j$ th  Monte Carlo realization of the payoff. 
\begin{algorithm}[H]
\caption{\label{algoMeanVar4}Global backward resolution algorithm, optimal control  and optimal variance calculation}
\begin{algorithmic}[1]
\For{$\nu \in \Qc_{N-1}$}
\For{$j\in \left[1,M\right]$}
\State$ \bar R(t_N ,S_{t_N}^{(j)}, \nu)  =  H^{(j)}$
\EndFor
\EndFor
\For{$i = N, 2$}
\For{$ k \xi \in \Qc_{i-2}$}
\State $(\nu^{(j)}_{i-1}(k))_{j=1,M} =$ OptimalControl$( \bar R(t_i,.,.), k, S_{t_{i-1}}, S_{t_i})$,
\For{$j\in \left[1,M\right]$}
\State $  \begin{array}{ll}  \bar R(t_{i-1}, S_{t_{i-1}}^{(j)}, k\xi) = &  \bar R(t_{i}, S_{t_{i}}^{(j)}, \nu_{i-1}^{(j)}(k)) -   \\
& \nu_{i-1}^{(j)}(k)  \Delta S_{i-1}^{(j)}
+  \lambda |\nu_{i-1}^{(j)}(k)-k \xi | S_{t_{i-1}}^{(j)}
\end{array} $
\EndFor
\EndFor
\EndFor
\State $P=\infty$, 
\For{$k = -\lfloor \frac{\bar m}{\xi} \rfloor, \ldots, \lfloor \frac{\bar l}{\xi} \rfloor $}
\State $\begin{array}{ll} \tilde P = &  \displaystyle{\sum_{j=1}^{M}}  \big( \bar R(t_1, S_{t_1}^{(j)}, k \xi)  - k \xi \Delta S_{0}^{(j)}
+   \lambda |k|  \xi  S_0 -  x )^2  \end{array}$
\If{ $ \tilde P < P$}
\State $\nu_0 = k \xi$, $P =  \tilde P$
\EndIf
\EndFor
\State $Var =  \frac{1}{M}\sum_{j=1}^M \big( \bar R(t_1,S_{t_1}^{(j)}, \nu_0) - \nu_0 \Delta S_0^{(j)}+\lambda |\nu_0| S_0 -x\big)^2$
\end{algorithmic}
\end{algorithm}
\section{Numerical results}
\subsection{The uncertainty models and the problem to solve}
We suppose that an electricity retailer has to face uncertainty on the load he has to provide for his customers.
We suppose that this  load   $D( t)$ is stochastic and follows the dynamic:
\begin{flalign}
\label{loadCurve}
D(t) = \hat D(t) + (D(u)-\hat D(u)) e^{-a_D ( t-u)} + \int_u^{ t} \sigma_D e^{-a_D( t-s)} dW^D_s, 
\end{flalign}
where $a_D$ is a mean-reverting coefficient, $\sigma_D$ the volatility of the process, $(W^D_t)_{t\le T}$ is a Brownian on $(\Omega, \Fc, \P)$ and $\hat D(u)$  is the average load seen on the previous years at the given date $u$. The equation \eqref{loadCurve} only states that the load curve $D(t)$ oscillates around an average value $\hat D(t)$ due to economic activity and to sensitivity to temperature.\\
We suppose that the retailer wants to hedge his position for a given date $T$ and that the future model 
is given by a one factor HJM model under the real world probability 
\begin{flalign}
F(t,T) = & F(0,T) e^{- \sigma_E^2 \frac{e^{-2 a_E(T-t)}-e^{-2 a_E T}}{4a_E} + e^{-a_E (T-t)} \hat W^E_t},\nonumber \\
\hat W^E_t = & \sigma_E \int_{0}^{t} e^{-a_E(t-s)} dW^E_s,
\label{futDyn}
\end{flalign}
where $F(t,T)$ is the forward curve seen at date $t$ for delivery at date $T$, $a_E$ the mean reverting parameter for electricity, $\sigma_E$ the volatility of the model and $(W^E_t)_{t\le T}$ a Brownian on $(\Omega, \Fc, \P)$ correlated to $W^D$ with a correlation $\rho$. The correlation is a priori negative, indicating that a high open position is a signal of a high available production or a low consumption driving the prices down. 
\begin{Remark}
The fact that the future price is modeled as a martingale is linked to the fact the risk premium is difficult to estimate and of second order compared to the volatility of the asset. This leads to strategies giving the same expected gains.
\end{Remark}
The SDE associated to the model \eqref{futDyn} is
\begin{flalign*}
dF(t,T) = \sigma_E e^{-a_E (T-t)} F(t,t) dW^E_t.
\end{flalign*}
Using  equation \ref{futDyn}, we get that
\begin{flalign*}
\E\left[\frac{F(t_i,T)^2}{F(t_{i-1},T)^2} | \Fc_{t_{i-1}}\right] = e^{\sigma_E^2 \frac{e^{-2 a_E(T-t_i)}-e^{-2 a_E (T-t_{i-1})}}{2a_E}},
\end{flalign*}
so that assumption \ref{assump1} is satisfied.\\
Similarly
\begin{flalign*}
\E\left[\frac{F(t_{i-1},T)^2}{F(t_{i},T)^2} | \Fc_{t_{i-1}}\right] = e^{3 \sigma_E^2 \frac{e^{-2 a_E(T-t_i)}-e^{-2 a_E (T-t_{i-1})}}{2a_E}},
\end{flalign*}
so that assumption \ref{assump2} is satisfied.
At last assumption \ref{assump3} is satisfied taking
\begin{flalign*}
\delta  = \max_{i=1}^N e^{- \sigma_E^2 \frac{e^{-2 a_E(T-t_i)}-e^{-2 a_E (T-t_{i-1})}}{2a_E}} < 1.
\end{flalign*}
The payoff of such a contract is then $H := D(T)F(T,T)$.
When there is no liquidity constraints, so supposing that the portfolio re-balancing is continuous, with no fees and constraints on volume and using the mean variance criterion (see remark \ref{remarkMeanVar}), the value of the contract $V = (V(t,D(t),F(t,T)))_{t\le T}$ and the hedging policy $\nu = (\nu(t,D(t),F(t,T)))_{t\le T}$ solve:
\begin{flalign}
\label{eq:contHedging}
(V,\nu)=  &\argmin_{(\tilde V, \tilde \nu) \in \R \times \Lc^2(F)} \E\left[  (D(T) F(T,T) -   \int_t^T \tilde \nu_s dF(s,T) - \tilde V)^2 | \Fc_t \right],
\end{flalign}
where $\Lc^2(F)$ is the set of predictable process  $v$  satisfying
\begin{flalign*}
\E\left[ \int_0^T v_t^2 \sigma_E^2 e^{-2a_E (T-t)} F(t,t)^2 dt\right] < \infty
\end{flalign*}
The demonstration of the following lemma is given in the appendix.
\begin{lemma}  \label{lemHedge}
  The optimal hedging policy solution of the equation  \eqref{eq:contHedging} is given by:
  \begin{flalign}
\nu(t, D(t), F(t,T))  =& \hat D(T) + (D(t)-\hat D(t)) e^{-a_D (T-t)}   \nonumber\\ 
& + \rho \left[e^{(a_E-a_D)(T-t)} \frac{\sigma_D}{\sigma_E} + \sigma_E \sigma_D \frac{1- e^{-(a_E+a_D)(T-t)}}{a_E+a_D} \right],
\label{couvMod}
  \end{flalign}
  and the  optimal variance  residual $Var^{opt}$ is given by:
  \begin{flalign}
  \label{optVarRes}
  Var^{opt} =& \sigma_D^2 (1- \rho^2) F(t,T)^2 \int_0^T e^{-2 a_D (T-s)} e^{\sigma_E^2 \frac{ e^{-2 a_E (T-s)}-e^{-2 a_E T}}{2 a_E}} ds.
\end{flalign}  
\end{lemma}

Traders often hedge their risks only using the sensitivity of the option with respect to the underlying getting a non optimal
hedging policy $\tilde \nu$  for the mean variance criterion:
\begin{flalign}
\label{nonOptCouv}
\tilde \nu(t, D(t), F(t,T))  = & \frac{\partial V}{\partial F}(t,D(t),F(t,T)) \\
=& \hat D(T) + (D(t)-\hat D(t)) e^{-a_D (T-t)}  + \rho \sigma_E \sigma_D \frac{1- e^{-(a_E+a_D)(T-t)}}{a_E+a_D}, \nonumber 
\end{flalign}
leading to a residual variance $Var^{cur}$ independent of the correlation given by:
\begin{flalign*}
  Var^{cur} = \sigma_D^2  F(t,T)^2  \int_0^T e^{-2 a_D (T-s)} e^{\sigma_E^2 \frac{ e^{-2 a_E (T-s)}-e^{-2 a_E T}}{2 a_E}} ds.
\end{flalign*}

\subsection{The test case}
We suppose that an important  producer  wants to hedge  an average monthly open  position of $9$ GW  three months before the beginning of delivery.
We suppose that the monthly product available on the market follows the dynamic \eqref{futDyn} with the following parameters:
$F(0,T)=40$ euros per MWh, $a_E=1.75$ and  $\sigma_E=20\%$  in annual.
The open position follows the dynamic given by equation \eqref{loadCurve} with parameters $a_D=19.8$ in annual, and  $\sigma_D=6240$ MW\@.
We suppose that the correlation $\rho$ between  $W^D$ and $W^E$ is equal to $-0.2$. \\
We suppose that the transaction costs are null, so we are only interested in the effect of the frequency of the hedge and the depth of the future market. 
We test the hedging strategies in term of variance   corresponding to:
\begin{itemize}
    \item the optimal analytic  given by equation \eqref{couvMod},
    \item the classical  tangent delta  given by equation
    \eqref{nonOptCouv},
    \item the numerical optimal solution obtained by algorithm \ref{algoMeanVar2}, \ref{algoMeanVar3} using some local basis functions as explained in paragraph \ref{sec:reg} and implemented with the StOpt library (\citeA{gevret2016stochastic}).
\end{itemize}
\begin{Remark}
Here the value function is a function of $F$ and $D$ so that two-dimensional regressions have to be used in the algorithms leading to the need of the specification of the number of mesh used for the representation of $F$ and $D$.
\end{Remark}
We first study the effect of the hedging frequency  with a market with infinite depth, then we  explore the effect of the finite depth of the market for the two algorithms using a high number of simulations and a high number of function basis.
At last, we  study the effect of the different parameters for the convergence of the algorithms. Because algorithm \ref{algoMeanVar2} needs to store the residual and the function value, it is more memory consuming  and runs slightly slower than the algorithm  \ref{algoMeanVar3} (less than $10\%$ of slowdown).\\
In all the tests, the global stock of hedging product available to sell is set to $12GW$ implying that there are no real global constraints for the availability of the product.
This stock of energy is discretized with a step of $100MW$  (so using $121$ position values) and at each trading date, the energy bought or sell on the market is discretized with the same step of $100MW$. 

\subsection{Market with infinite depth}
In this section, we take the following parameters for the solver: $8 \times 8$ meshes for the two-dimensional regressions. In optimization, we take $400000$ trajectories to calculate the regressions.
On table \ref{VarInfDeepth}, we  test the influence of the hedging frequency. We give  the variance obtained  by taking the average of 10 runs of  algorithms \ref{algoMeanVar2} and  \ref{algoMeanVar3}. Using an estimation $\hat  \sigma$ of the standard deviation of the different runs, we also give the value $\frac{\sigma}{\sqrt{10}}$ as an estimation of the confidence in the result.
Using the control obtained in an optimization run, we also simulate the variance obtained  out of the sample  by the numerical optimal strategy, the classical tangent delta and the optimal analytic using $1e6$ particles in table \ref{VarInfDeepthSim}.
The optimal variance given by equation \eqref{optVarRes} for a continuous hedging is $7.8628e+14$. Without hedge, the variance of the portfolio is equal to $1.114e+15$.
We notice that the averaged value during the optimization part is very close to the value obtained using  out of the sample trajectories. Numerically with $1e6$ trajectories, we even get a variance slightly below the analytical one showing that the convergence of the estimated variance is slow and that a very high number of samples are necessary to get a very good estimate.\\
As shown in table \ref{VarInfDeepthSim},  by using a time step of 2 weeks (a number of hedging dates equal to 8), the  numerical hedging residual is very close to the one obtained using the optimal analytical hedge and both values are very close to the values obtained by continuous hedging.
Besides,  we see that  the numerical strategy works better than the optimal analytic strategy  only when  3 or 4 hedging dates are used.

\begin{table}[H]
\centering
 \begin{tabular}{|c|c|c|c|}  \hline
   Number of hedging dates     & 3    &  4    &  8    \\ \hline
   Averaged variance           & 7.953e+14     &  7.9129e+14     &  7.851e+14     \\ \hline
   $\frac{\sigma}{\sqrt{10}}$  & 5.449e+11     &  6.328e+11     &  2.315e+11    \\ \hline
  \end{tabular}
\caption{\label{VarInfDeepth}Variance  obtained for infinite market depth with algorithm \ref{algoMeanVar3}.}
\end{table}

\begin{table}[H]
\centering
 \begin{tabular}{|c|c|c|c|}  \hline
   Number of hedging dates     & 3    &  4    &  8    \\ \hline
   Numerical         &    7.952e+14       &  7.9106e+14  &   7.853e+14 \\ \hline
   Optimal analytic  &   8.0843e+14     &  7.99811e+14     &  7.852e+14     \\ \hline
   Classical hedge    &  8.1905e+14    &   8.1854e+14    & 8.157e+14      \\ \hline
  \end{tabular}
\caption{\label{VarInfDeepthSim}Out of the sample variance  observed  for the  different strategies  for infinite market depth with algorithm \ref{algoMeanVar3}.}
\end{table}
On table \ref{VarInfDeepthAlgo1} and \ref{VarInfDeepthSimAlgo1}, we give similar results obtained by the algorithm \ref{algoMeanVar2}.
Taking 3 hedging dates, we get the same results as with algorithm \ref{algoMeanVar3} (differences exist but appears on the fifth digit).
As we increase the number of hedging dates, some differences appear between the two algorithms but the results are always very similar.
\begin{table}[H]
\centering
 \begin{tabular}{|c|c|c|c|}  \hline
   Number of hedging dates     & 3    &  4    &  8    \\ \hline
   Averaged variance           & 7.953e+14     &  7.9166e+14     &   7.8664e+14   \\ \hline
   $\frac{\sigma}{\sqrt{10}}$  & 5.449e+11     &  6.31809e+11    &  2.313e+11  \\ \hline
  \end{tabular}
\caption{\label{VarInfDeepthAlgo1}Variance  obtained for infinite market depth with algorithm \ref{algoMeanVar2}.}
\end{table}

\begin{table}[H]
\centering
 \begin{tabular}{|c|c|c|c|}  \hline
   Number of hedging dates     & 3    &  4    &  8    \\ \hline
   Numerical         &    7.952e+14       &  7.9069e+14   &  7.8797e+14 \\ \hline
  \end{tabular}
\caption{\label{VarInfDeepthSimAlgo1}Out of the sample variance  observed for infinite market depth  for the optimal strategy with algorithm \ref{algoMeanVar2}.}
\end{table}

\subsection{Market with finite depth}
In this section, we suppose that at each hedging date the power of energy available is limited to $1200W$.
We keep the same parameters for the models and the resolution as in the previous section. The optimal analytic strategies and classical hedging strategies are obtained by clipping the analytical strategies
in order to  respect the finite depth constraints.\\
The results are given in tables  \ref{VarFinDeepth} and \ref{VarFinDeepthSim} for algorithm \ref{algoMeanVar3}. For 3 and 4 hedging dates, the among of energy to sell on the future market is too low and the optimal strategy is always to sell at the maximum.
For $8$ and $13$ hedging dates, the numerical strategies and the optimal analytic give the same variance.
\begin{table}[H]
\centering
 \begin{tabular}{|c|c|c|c|c|}  \hline
   Number of hedging dates     & 3    &  4    &  8  &   13 \\ \hline
   Averaged variance           & 9.82228e+14   &  9.51851e+14   & 8.84079e+14  & 8.3508e+14 \\ \hline
   $\frac{\sigma}{\sqrt{10}}$  & 6.19103e+11   &  6.36506e+11  & 2.62856e+11  &  4.34744e+11 \\ \hline
  \end{tabular}
\caption{\label{VarFinDeepth}Variance  obtained for finite market depth with algorithm \ref{algoMeanVar3}.}
\end{table}

\begin{table}[H]
\centering
 \begin{tabular}{|c|c|c|c|c|}  \hline
   Number of hedging dates     & 3     &  4    &  8  &  13  \\ \hline
   Numerical           &  9.81158e+14       &  9.49984e+14     &  8.82166e+14  &    8.36344e+14    \\ \hline
   Optimal analytic  &  9.81158e+14       &  9.49984e+14     &  8.84816e+14  &  8.3635e+14 \\ \hline
   Classical hedge    &  9.81158e+14      &  9.49984e+14     &  8.96696e+14  &  8.67346e+14 \\ \hline
  \end{tabular}
\caption{\label{VarFinDeepthSim}Out of the sample variance  observed  for finite market depth for the  different strategies with algorithm \ref{algoMeanVar3}.}
\end{table}

Taking a correlation of $-0.2$, we see no advantage of using the numerical approach instead of using the analytical solution.
In the table \ref{VarFinDeepthCorr}, keeping  the number a hedging dates equal to $8$, we give the variance obtained with algorithm \ref{algoMeanVar3} depending on the correlation and in table  \ref{VarFinDeepthSimCorr}, we give the results obtained  by the different strategies using out of the sample trajectories.  
\begin{table}[H]
\centering
 \begin{tabular}{|c|c|c|}  \hline
  Correlation     &  -0.4    &  -0.6    \\ \hline
   Averaged variance           & 7.84187e+14   & 6.44759e+14   \\ \hline
   $\frac{\sigma}{\sqrt{10}}$  & 2.65632e+11   & 2.61286e+11    \\ \hline
  \end{tabular}
\caption{\label{VarFinDeepthCorr}Variance  for different  correlations obtained for finite market depth with $8$ hedging dates with algorithm \ref{algoMeanVar3}.}
\end{table}
\begin{table}[H]
\centering
 \begin{tabular}{|c|c|c|}  \hline
   Correlation     & -0.4    &  -0.6   \\ \hline
   Numerical           & 7.84736e+14     & 6.45161e+14    \\ \hline
   Optimal analytic  &  7.97782e+14    &  6.9845e+14   \\ \hline
   Classical hedge   & 8.57933e+14    &  8.17449e+14    \\ \hline
  \end{tabular}
\caption{\label{VarFinDeepthSimCorr}Out of the sample variance  observed  for finite market depth with $8$ hedging dates  for the  different strategies with algorithm \ref{algoMeanVar2}.}
\end{table}
 Results obtained with the
algorithm \ref{algoMeanVar2} for different correlations  in the optimization part are given in table \ref{VarFinDeepthCorr1}. Once again, results obtained by the two algorithms are very close.
\begin{table}[H]
\centering
 \begin{tabular}{|c|c|c|}  \hline
  Correlation     &  -0.4    &  -0.6    \\ \hline
   Averaged variance           &  7.84423e+14   & 6.45014e+14  \\ \hline
   $\frac{\sigma}{\sqrt{10}}$  & 2.45677e+11   & 2.63838e+11   \\ \hline
  \end{tabular}
\caption{\label{VarFinDeepthCorr1}Variance  for different  correlations obtained for finite market depth with $8$ hedging dates with algorithm \ref{algoMeanVar2}.}
\end{table}
The results obtained clearly show that clipping the analytical continuous optimal hedging policy may be far from  optimal  especially for high correlations. Results are confirmed
on the figure \ref{AProfFaibleD265AverageDispersionVF2Cor02} where the different strategies are plot on the same trajectories for a correlation of $-0.4$. The numerical optimizer is able to anticipate that the position at maturity should still be open and  the position is kept longer than the one obtained  with  the optimal analytical strategy.

\begin{figure}[H]
\begin{minipage}[b]{0.49\linewidth}
  \centering
 \includegraphics[width=\textwidth]{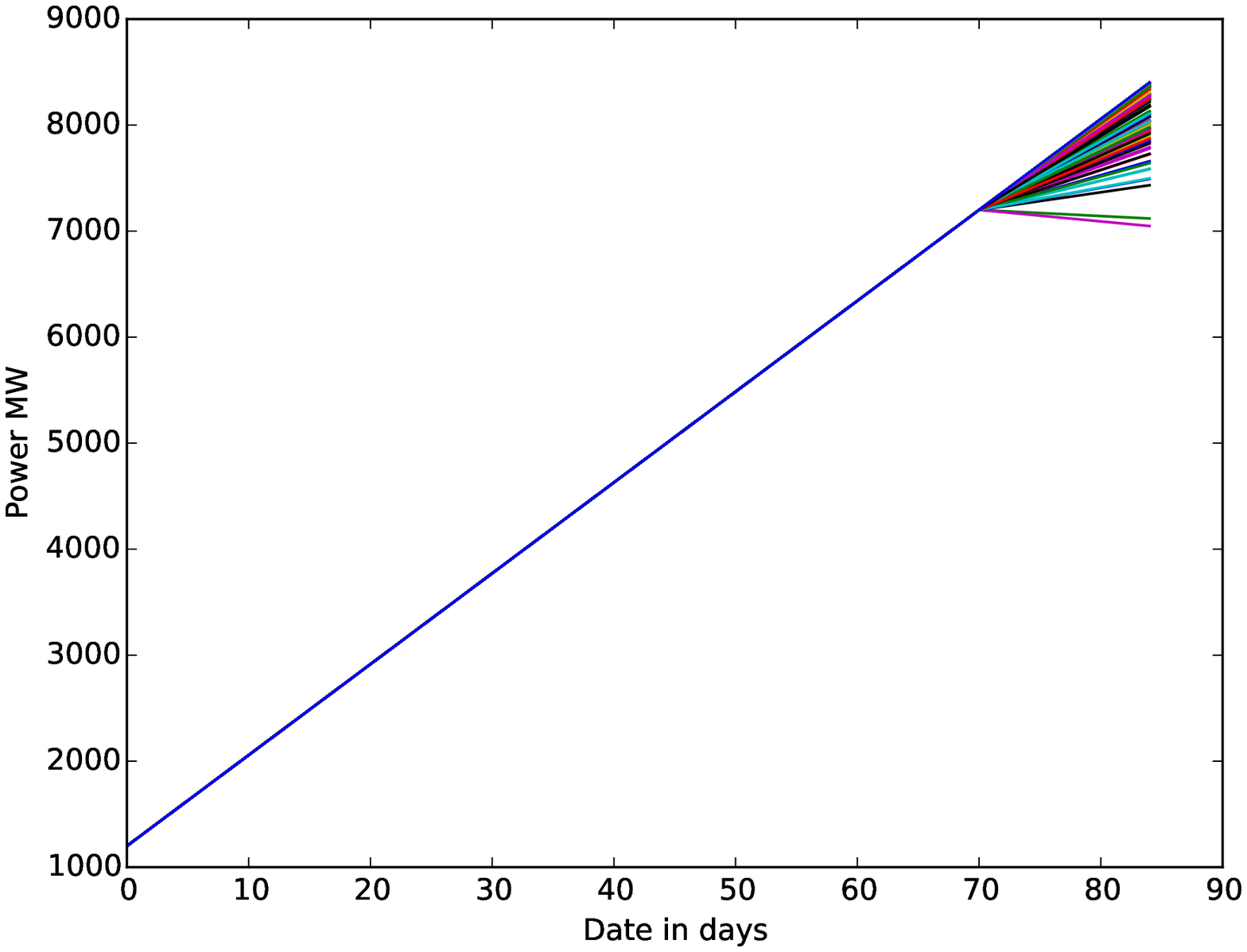}
 \caption*{Delta tangent.}
 \end{minipage}
 \begin{minipage}[b]{0.49\linewidth}
  \centering
 \includegraphics[width=\textwidth]{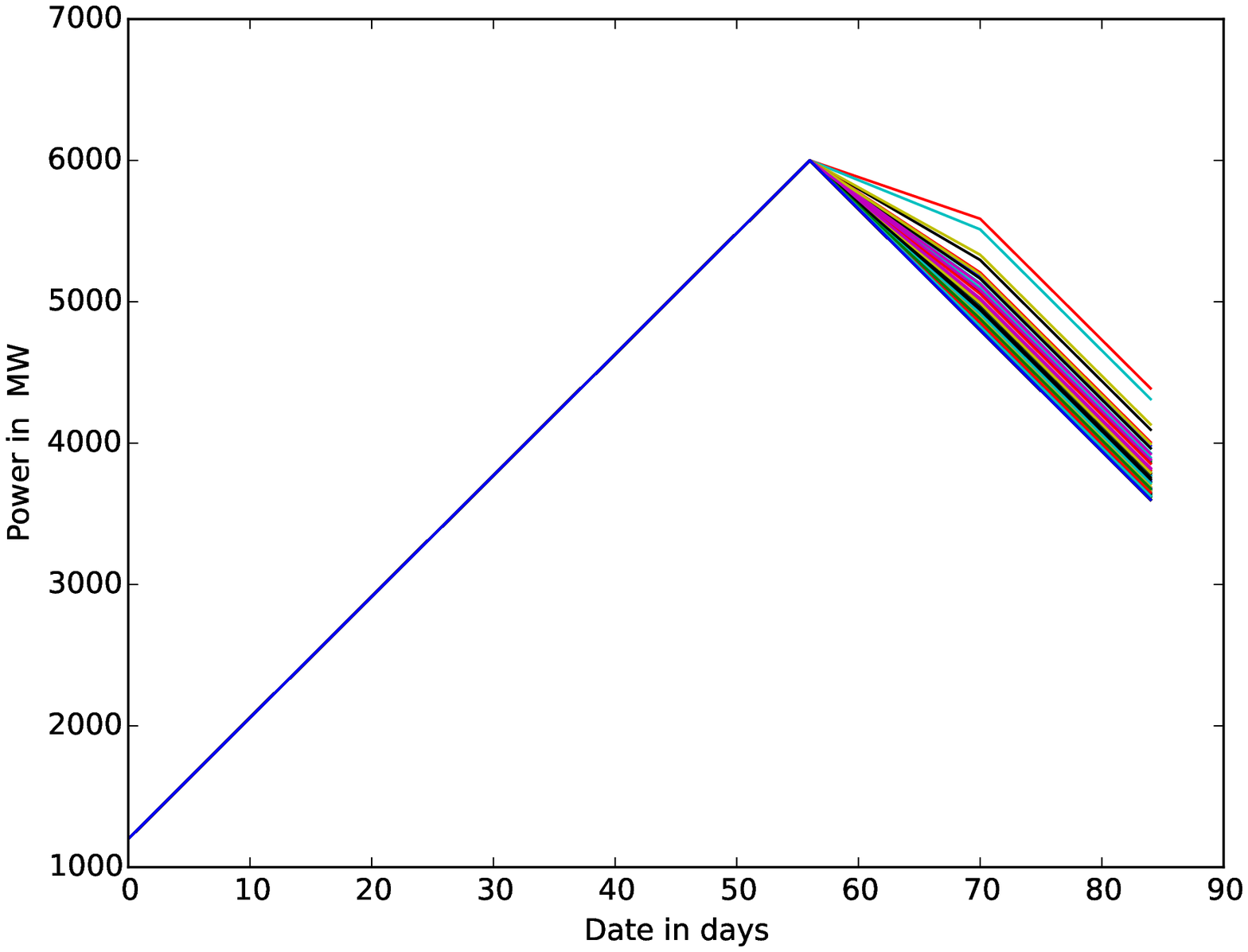}
 \caption*{Analytic optimal.}
 \end{minipage}
\begin{minipage}[b]{\linewidth}
  \centering
 \includegraphics[width=0.5\textwidth]{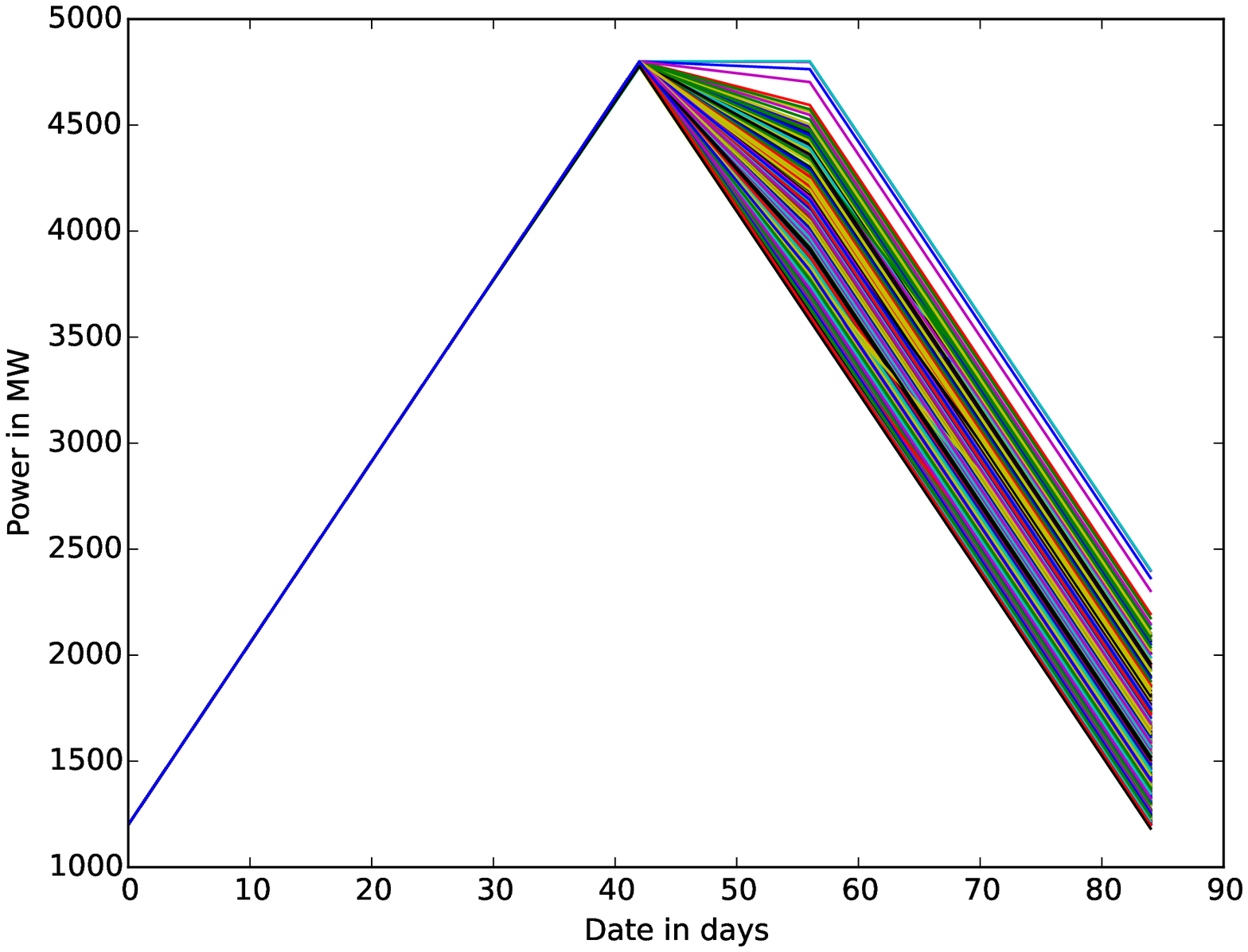}
 \caption*{Numerical.}
 \end{minipage}
\caption{\label{AProfFaibleD265AverageDispersionVF2Cor02} 
Hedging simulations in MW  for finite market depth with $8$ hedging dates and  a correlation of -0.4.}
\end{figure}

\section{Numerical convergence of the two algorithms}
In this section, we  study the influence on the results of the number of mesh and the number of trajectories taken in optimization.
To this end, we take the infinite market depth test case with a number of hedging dates equal to $8$. 
As shown in lemma \ref{lemHedge}, in the continuous case, the conditional variance is quadratic  in the future value and is  independent  of  the open position, while
the optimal control is only affine with the open position.\\
As a reference, we take the average of $100$ runs of each algorithm using $12 \times 12$ meshes and one million trajectories in optimization.
By construction, the local regression algorithm tries to give a partition such that the  number of trajectories belonging to a cell is constant. 
As a rule of the thumb,  the number of trajectories on each mesh is kept  roughly constant  (so here equal to $7000$) while changing the number of mesh used.
First, we take a number of mesh as $n \times n$ which is clearly not optimal: the number of mesh in the $F$ and the  $D$ direction should be taken different and dependent on the volatilities and mean reverting coefficients.
On the table \ref{meshImpact}, we give the average of $100$ runs obtained with algorithm \ref{algoMeanVar2} and \ref{algoMeanVar3} for different numbers of meshes.
We also plot the standard deviation of the calculation $\hat \sigma$ divided by $\sqrt{100}$ as an indicator of the error.
The bias due to a small number of mesh is rather small for both algorithms. Off course using a small number of mesh leads to a rather high variance
of the result obtained.
\begin{table}[H]
\centering
 \begin{tabular}{|c|c|c|c|c|c|}  \hline
   Number of meshes      & $1 \times 1$     &  $2 \times 2$  & $4 \times 4$  & $6 \times 6$    \\ \hline
   Average  algo. \ref{algoMeanVar2} &  7.894e+14   &  7.865e+14   &  7.876e+14    & 7.877e+14        \\ \hline
   $\frac{\hat \sigma}{\sqrt{100}}$   algo. \ref{algoMeanVar2} & 1.525e+12 & 6.812e+11  & 3.602e+11  & 2.436e+11  \\ \hline
   Average   algo. \ref{algoMeanVar3} & 7.869e+14  & 7.847e+14   & 7.861e+14   & 7.862e+14   \\ \hline
   $\frac{\hat \sigma}{\sqrt{100}}$  algo. \ref{algoMeanVar3}  & 1.524e+12 & 6.791e+11& 3.577e+11 & 2.445e+11 \\ \hline
 \end{tabular}
 \newline
 \bigskip
 \newline
  \begin{tabular}{|c|c|c|c|}  \hline
   Number of meshes      &   $8 \times 8$ & $10 \times 10$ & $12 \times 12$   \\ \hline
   Average  algo. \ref{algoMeanVar2} &  7.876e+14   & 7.875e+14  &  7.875e+14    \\ \hline
   $\frac{\hat \sigma}{\sqrt{100}}$   algo. \ref{algoMeanVar2} &  1.888e+11  & 1.21038e+11  & 1.227e+11   \\ \hline
   Average   algo. \ref{algoMeanVar3} &  7.860e+14 & 7.860e+14  &  7.860e+14\\ \hline
   $\frac{\hat \sigma}{\sqrt{100}}$  algo. \ref{algoMeanVar3}  & 1.906e+11 & 1.221e+11 & 1.242e+11\\ \hline
  \end{tabular}
\caption{\label{meshImpact} Convergence study of the two algorithms with the number of mesh keeping the same number of trajectories (roughly 7000) on each mesh.}
\end{table}
In table \ref{simImpact}, we choose to keep a number $8 \times 8$ mesh and increase the number of trajectories per mesh.
\begin{table}[H]
\centering
 \begin{tabular}{|c|c|c|c|c|c|}  \hline
   Number of trajectories     &  50000 &  100000 & 200000 & 440000    &  1000000 \\ \hline
   Average  algo. \ref{algoMeanVar2} &  7.8399e+14 & 7.8643e+14 & 7.8700e+14 &   7.8759e+14   & 7.87729e+14            \\ \hline
   $\frac{\hat \sigma}{\sqrt{100}}$   algo. \ref{algoMeanVar2} &  4.624e+11& 3.121e+11 & 2.621e+11 &1.888e+11   & 1.23053e+11   \\ \hline
   Average   algo. \ref{algoMeanVar3} &  7.7126e+14 &   7.8007e+14 &  7.8378e+14 &7.86048e+14 & 7.87052e+14  \\ \hline
   $\frac{\hat \sigma}{\sqrt{100}}$  algo. \ref{algoMeanVar3}  &  4.686e+11 & 3.141e+11 & 2.664e+11& 1.906e+11  & 1.23911e+11   \\ \hline
 \end{tabular}
 \caption{\label{simImpact} Using $8 \times 8$  meshes, convergence  study of the two algorithms increasing the number of particles.}
 \end{table}
Indeed, both algorithm give similar results  for a high number of particles per mesh. However, for a low number of particles per mesh, algorithm \ref{algoMeanVar2} presents a lower bias than algorithm  \ref{algoMeanVar3}.\\
At last we want to check that the control is well calculated. In the general case the optimal controls is in dimension three but in the continuous case the formula of lemma \ref{lemHedge}
gives a mono-dimensional representation of the optimal control.\\
The optimal control has been calculated with algorithm \ref{algoMeanVar3} and projected on the local function basis.
Then we can reconstruct the optimal position calculated  by the algorithm and compare the solution with  the formula in lemma  \ref{lemHedge}.
The figure \ref{control} shows that the  numerical optimal control is affine  and that even if the variance obtained with 8 hedging dates is very close to the continuous one, the  optimal control is still far away from the continuous one. Increasing the number of hedging dates permits to be closer to the continuous optimal control.

\begin{figure}[H]
\begin{minipage}[b]{0.49\linewidth}
  \centering
 \includegraphics[width=\textwidth]{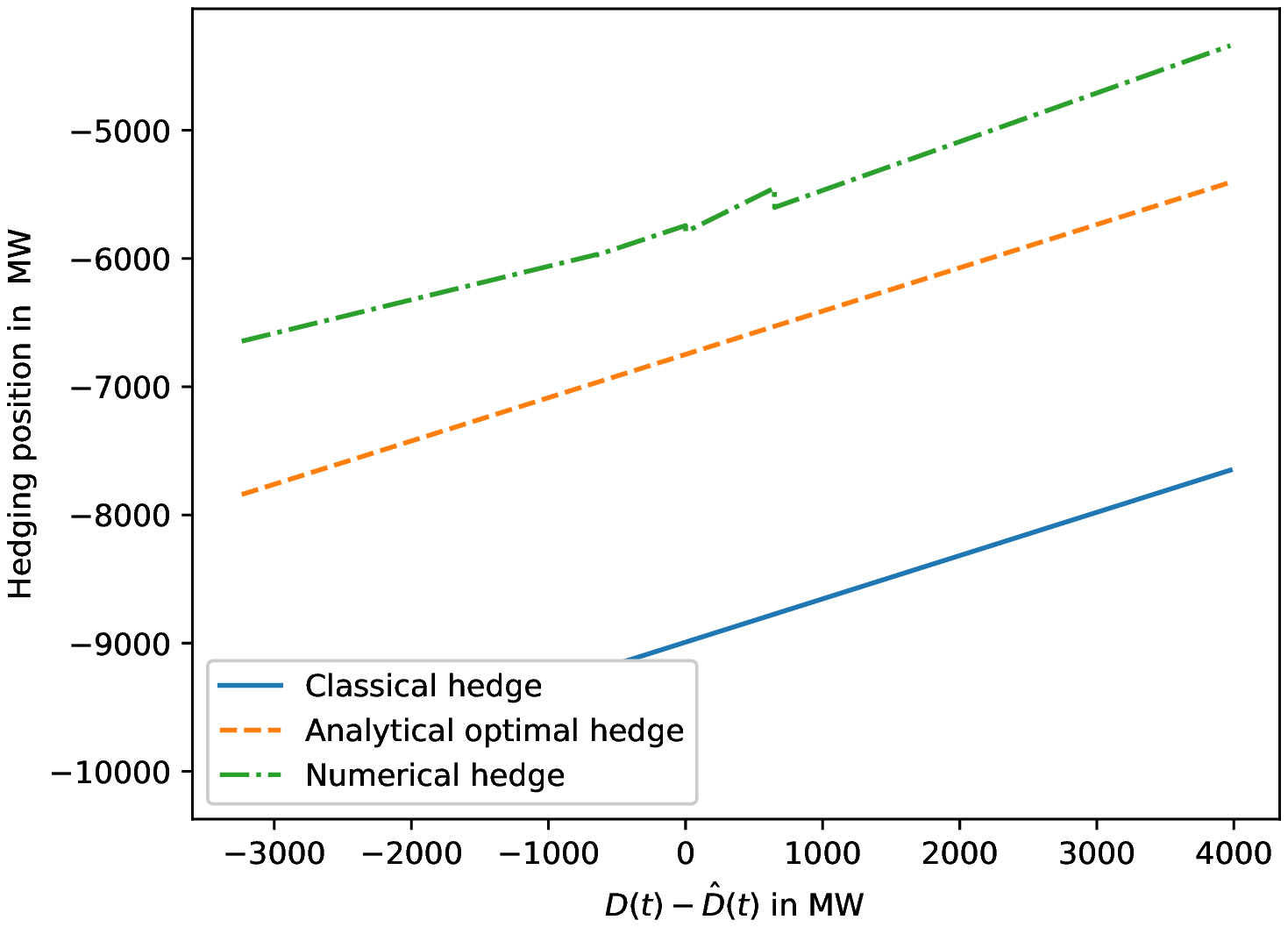}
 \caption*{8 hedging dates.}
 \end{minipage}
 \begin{minipage}[b]{0.49\linewidth}
  \centering
 \includegraphics[width=\textwidth]{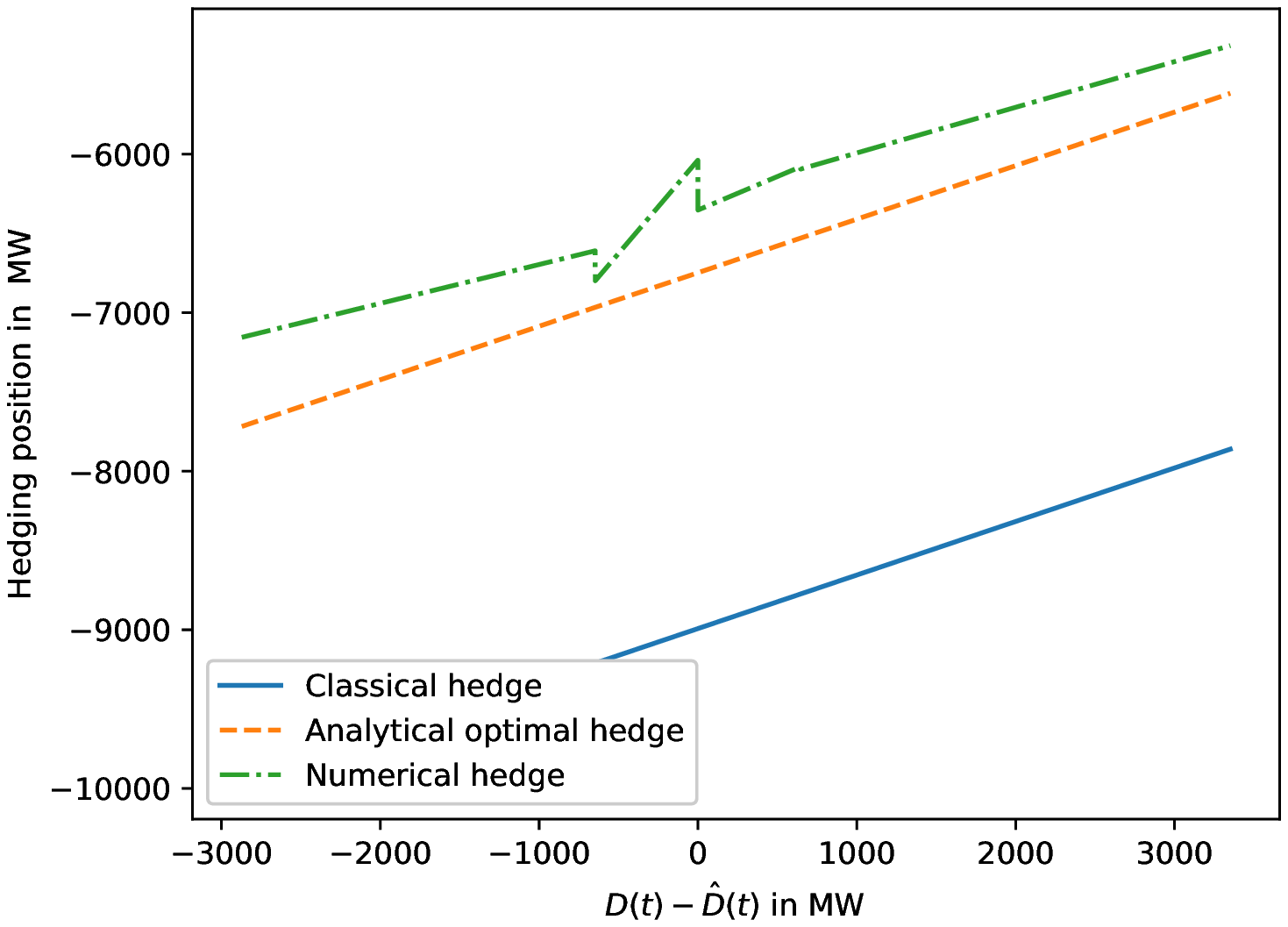}
 \caption*{14 hedging dates.}
 \end{minipage}
 \begin{minipage}[b]{\linewidth}
  \centering
 \includegraphics[width=0.5\textwidth]{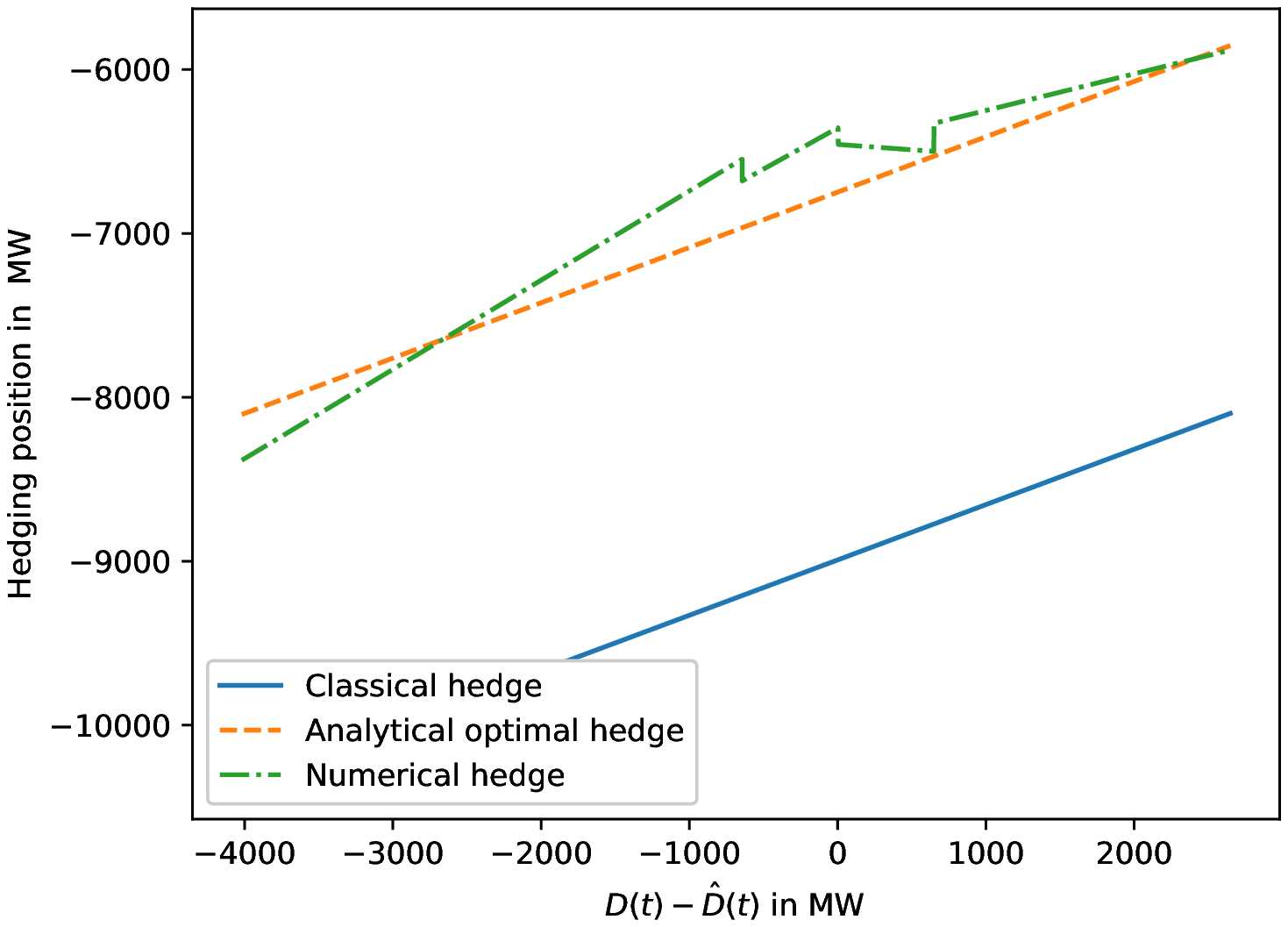}
 \caption*{26 hedging dates.}
 \end{minipage}
\caption{\label{control} 
Comparison of numerical optimal position, optimal analytic position, and classical hedge obtained 20 days before delivery for different number of hedging dates using $4$ basis functions for the open position.}
\end{figure}

\section{Conclusion}

 In the case of mean variance hedging of the wealth portfolio some effective algorithms has been developed to find the optimal strategy taking into account the transaction costs, the limited availability of the hedging products, the fact that hedging is only achieved at discrete dates.\\
We have shown on a realistic case in energy market that taking in account the reality of the market depth  has an important impact on the efficiency of the hedging strategy. \\
This algorithm could be extended to non symmetric risk measure to take into account the fact that managers wants to favor gains as in \citeA{gobet2018option}.

\section{Appendix}
We give the proof of lemma \ref{lemHedge}: \\

The solution of the  problem \eqref{couvMod} is known to be given in the case of martingale assets by the Galtchouk-Kunita-Wananabe decomposition and 
\begin{flalign*}
V(t,D(t),F(t,T)) = \E\left[ D(T)F(T,T)| \Fc_t\right],
\end{flalign*}
which can be calculated as follows:
\begin{flalign*}
V(t,D(t),F(t,T)) = & F(t,T) \left[\hat D(T) + (D(t)-\hat D(t)) e^{-a_D (T-t)} + \right. \nonumber \\
&  \left. \quad \quad \rho \sigma_E \sigma_D \frac{1- e^{-(a_E+a_D)(T-t)}}{a_E+a_D}\right].
\end{flalign*}
The value function $V$ being a martingale, using Ito lemma we get:
\begin{flalign*}
D(T) F(T,T) = & V(0,D(0), F(0,T)) + \int_{0}^T \frac{\partial V}{\partial D}(s,D(s),F(s,T)) \sigma_D d W^D_s  + \\
& \int_{0}^T \frac{\partial V}{\partial F}(s,D(s),F(s,T)) \sigma_E F(s,T) e^{-a_E (T-s)} dW^E_s, 
\end{flalign*}
so using  $W^D_t  = \rho W^E_t + \sqrt{1-\rho^2}  \hat W^D_t$ with $ \hat W^D_t$ orthogonal to $W^E_t$ in $\Lc^2$:
\begin{flalign}
D(T) F(t,T) = & V(0,D(0), F(0,T)) + \int_{0}^T  \left[ \frac{\partial V}{\partial F}(s,D(s),F(s,T)) + \nonumber \right. \\
&  \left. \quad \quad   \rho \frac{\sigma_D e^{a_E(T-s)}}{\sigma_E F(s,T)} \frac{\partial V}{\partial D}(s,D(s),F(s,T))\right]  d F(s,T) \nonumber \\
& +  \sqrt{1- \rho^2}  \frac{\partial V}{\partial D}(s,D(s),F(s,T)) \sigma_D d \hat W^D_s. 
\label{hedge}
\end{flalign}
The second part in the previous integral represents the non hedgable part of the asset and the optimal hedge is given by:
\begin{flalign}
  \nu(t, D(t), F(T,T))  = & \frac{\partial V}{\partial F}(t,D(t),F(t,T)) + \rho \frac{\sigma_D e^{a_E(T-s)}}{\sigma_E F(t,T)} \frac{\partial V}{\partial D}(t,D(t),F(t,T)).
  \label{hedgeOpt}
\end{flalign}
Introducing the forward tangent process
\begin{flalign*}
Y_t^T = e^{ -\frac{V(t,T)}{2} +  e^{-a_E (T-t)}  \hat W^E_t },
\end{flalign*}
we get 
\begin{flalign}
  \label{DV}
 \frac{\partial V}{\partial F}(t,D(t),F(t,T)) = & \E( D(T) Y^T_T |  \Fc_t )/Y_t^T, \nonumber \\ 
 \frac{\partial V}{\partial D}(t,D(t),F(t,T)) = & e^{-a_D (T-t)} F(t,T), 
\end{flalign}
so that
\begin{flalign}
  \label{DVDF}
\frac{\partial V}{\partial F}(t,D(t),F(t,T))& =  \hat D(T) + (D(t)-\hat D(t)) e^{-a_D (T-t)} + \nonumber \\
&  \rho \sigma_E \sigma_D \frac{1- e^{-(a_E+a_D)(T-t)}}{a_E+a_D}.
\end{flalign}
Plugging equation \eqref{DVDF} in equation \eqref{DV} and equation \eqref{DV} in equation \eqref{hedgeOpt} gives the optimal hedging strategy.
The residual risk is straightforward.

\section{Acknowledgements and Declarations of Interest}
{\bf Acknowledgements}:
This work has benefited from the financial support of the ANR Caesars.\\
{\bf Declarations of Interest}:
The author reports no conflicts of interest. The author alone is responsible for the content and writing of the paper.

\bibliographystyle{apacite}
\bibliography{jcf}

\end{document}